\documentclass[algorithms,article,submit,moreauthors,pdftex,12pt,a4paper]{mdpi} 
\lastpage{x}
\doinum{10.3390/------}
\pubvolume{xx}
\pubyear{2012}
\history{Received: xx / Accepted: xx / Published: xx}

\usepackage{graphicx}
\usepackage{subfig}
\usepackage{algorithm}
\usepackage{algorithmic}
\usepackage{natbib}
\usepackage{amssymb,amsmath,comment,amsthm}
\newtheorem{theorem}{Theorem}

\Title{New heuristics on Rooted Triplet Consistency\thanks{An extended abstract of this article has appeared in Proceedings of Annual International Conference on Bioinformatics and Computational Biology (BICB 2011) in Singapore.}
}

\Author{Soheil Jahangiri Tazehkand $^{1,3,\star}$, Seyed Naser Hashemi $^{1}$ and Hadi Poormohammadi $^{2}$}

\address{%
$^{1}$ Amirkabir University of Technology (Tehran Polytechnic), 424 Hafez Ave, Tehran, Iran\\
$^{2}$ Shahid Beheshti University, Evin, Tehran, Iran\\
$^{3}$ Bioinformatics Group, School of Computer Science, Institute for Research in Fundamental Sciences (IPM), Niavaran, Tehran, Iran}
\corres{s.jahangiri@aut.ac.ir, Tel. +98 936 9220139.}

\abstract{Rooted triplets are becoming one of the important types of input for reconstructing rooted phylogenies. A rooted triplet is a phylogenetic tree on three leaves and shows the evolutionary relationship of the corresponding three species. In this paper, we investigate the problem of inferring the maximum consensus evolutionary tree from a set of rooted triplets. The mentioned problem is known to be APX-hard. We present two new heuristic algorithms. For a given set of \emph{m} triplets on \emph{n} species, the \emph{FastTree} algorithm runs in $O(m + \alpha(n)n^{2})$ time, where $\alpha(n)$ is functional inverse of Ackermann's function. This is faster than any other previously known algorithms, although, the outcome is less satisfactory. The \emph{BPMTR} algorithm runs in $O(mn^{3})$ time and in average performs better than any other previously known algorithms for this problem.}


\begin{document}

\section{Introduction}
\label{sec:1}

After publication of Charles Darwin's book \emph{On the origin of species; By means of natural selection}, the theory of evolution was widely accepted. Since then remarkable developments in evolutionary studies brought the scientists to the Phylogenetics, a field that studies the biological or the morphological data of species to output a mathematical model such as a tree or a network representing the evolutionary interrelationship of species and the process of their evolution. Besides, Phylogenetics is not only limited to the biology but may also arise anywhere that the concept of evolution appears. For example, a recent study in evolutionary linguistic employs phylogeny inference to clarify the origin of Indo-European language family\cite{Bouckaert2012}.
Several approaches have been introduced to infer evolutionary relationships \cite{Felsenstein2004}. Amongst those, well known approaches are character based methods (e.g., Maximum Parsimony), distance based methods (e.g., Neighbor Joining and UPGMA) and quartet based methods (e.g., QNet). Recently, rather new approaches namely triplet based methods have been introduced. Triplet based methods output rooted trees and networks due to the rooted nature of triplets. A rooted triplet is a rooted unordered leaf labeled binary tree on three leaves and shows the evolutionary relationship of the corresponding three species. Triplets can be obtained accurately using a maximum likelihood method such as the one introduced by \citet{Chor2001} or Sibley-Ahlquist-style DNA-DNA hybridization experiments \cite{Kannan1996}. Indeed, we expect highly accurate results from triplet based methods. However, sometimes due to experimental errors or some biological events such as hybridization (recombination) or horizontal gene transfer it is not possible to reconstruct a tree that satisfies all of the input constraints (triplets). There are two approaches to overcome this problem. The first approach is to employ a more complex model such as network which is the proper approach when the mentioned biological events have actually happened. The second approach tries to reconstruct a tree satisfying as many input triplets as possible. This approach is more useful when the input data contains error. The latter approach forms the subject of this paper. In the next section we will provide necessary definitions and notations. Section 3 contains an overview of previous results. We will present our algorithms and experimental results in section 4. Finally, in section 5 open problems and ideas for further improvements are discussed.
\section{Preliminaries}
\label{sec:2}

An \emph{evolutionary tree (phylogenetic tree)} on a set \emph{S} of \emph{n} species, \begin{math}|S|=n\end{math}, is a binary, rooted\footnote{More precisely speaking, an evolutionary tree can also be unrooted, however triplet based methods output rooted phylogenies.}, unordered tree in which leaves are distinctly labeled by members of \emph{S} (see Fig.~\ref{fig:1a}). A \emph{rooted triplet} is a phylogenetic tree on three leaves. The unique triplet on leaves {\emph{x}, \emph{y}, \emph{z}} is denoted by (({\emph{x}, \emph{y}}), \emph{z}) or \emph{xy}\textbar\emph{z}, if the lowest common ancestor of \emph{x} and \emph{y} is a proper descendant of the lowest common ancestor of \emph{x} and \emph{z}, or equivalently the lowest common ancestor of \emph{x} and \emph{y} is a proper descendant of lowest common ancestor of \emph{y} and \emph{z} (see Fig.~\ref{fig:1b}). A triplet \emph{t} (e.g., \emph{xy}\textbar\emph{z}) is \emph{consistent} with a tree \emph{T}  (or equivalently \emph{T} is consistent with \emph{t}) if \emph{t} is an embedded subtree of \emph{T}. It means \emph{t} can be obtained from \emph{T} by a series of edge contractions (i.e., if in \emph{T} the lowest common ancestor of \emph{x} and \emph{y} is a proper descendant of the lowest common ancestor of \emph{x} and \emph{z}). We also say \emph{T} \emph{satisfies} \emph{t}, if \emph{T} is consistent with \emph{t}. The tree in Fig.~\ref{fig:1a} is consistent with the triplet in Fig.~\ref{fig:1b}. A phylogenetic tree \emph{T} is consistent with a set of rooted triplets if it is consistent with every triplet in the set. We call two leaves \emph{siblings} or \emph{cherry} if they share the same parent. For example, \{\emph{x}, \emph{y}\} in Fig.~\ref{fig:1a} form a cherry.
\begin{figure}
  \caption{Example of a phylogenetic tree and a consistent triplet}
  \label{fig:1}
  \begin{center}
    \subfloat[A phylogenetic tree]
    {
        \label{fig:1a}
        \includegraphics[scale=0.75]{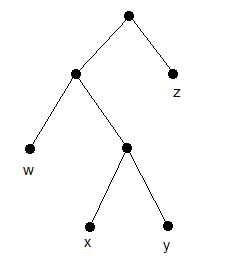}
    }
    \subfloat[The triplet {\emph{x}\emph{y}\textbar\emph{z}}]
    {
        \label{fig:1b}
        \includegraphics[scale=0.75]{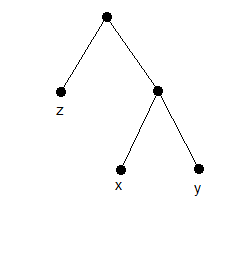}
    }
  \end{center}
\end{figure}

A set of triplets \emph{R} is called \emph{dense} if for each set of three species \{\emph{x}, \emph{y}, \emph{z}\}, \emph{R} contains at least one of three possible triplets \emph{xy}\textbar\emph{z}, \emph{xz}\textbar\emph{y} or \emph{yz}\textbar\emph{x}. If \emph{R} contains exactly one triplet for each set of three species, it is called \emph{minimal dense}, and if it contains every possible triplet it is called \emph{maximal dense}. Now we can define the problem of reconstructing an evolutionary tree from a set of rooted triplets. Suppose \emph{S} is a finite set of species of cardinality \emph{n} and \emph{R} is a finite set of rooted triplets of cardinality \emph{m} on \emph{S}. The problem is to find an evolutionary tree leaf-labeled by members of \emph{S} which is consistent with the maximum number of rooted triplets in \emph{R}. This problem is called \emph{Maximum Rooted Triplets Consistency (MaxRTC)} problem \cite{Byrka2010} or \emph{Maximum Inferred Local Consensus Tree (MILCT)} problem \cite{Jansson2001}. This problem is NP-hard (see section ~\ref{sec:3}) which means no polynomial-time algorithm can be found to solves the problem optimally unless P=NP. For this problem and similar problems, one might search for polynomial-time algorithms that produce approximate solutions. We call an algorithm an \emph{approximation algorithm} if its solution is guaranteed to be within some factor of optimum solution. In contrast, \emph{heuristics} may produce good solutions but do not come with a guarantee on their quality of solution. An algorithm for a maximization problem is called an \begin{math}\alpha-approximation\end{math} algorithm, for some \begin{math}\alpha>1\end{math}, if for any input the output of algorithm be at most \begin{math}\alpha\end{math} times worse than the optimum solution. The factor \begin{math}\alpha\end{math} is called \emph{approximation factor} or \emph{approximation ratio}.
\section{Related works}
\label{sec:3}

\citet{Aho1981} investigated the problem of constructing a tree consistent with a set of rooted triplets for the first time. They designed a simple recursive algorithm which runs in \begin{math}O(mn)\end{math} time and returns a tree consistent with all of the given triplets if at least one tree exists. Otherwise, it returns null. Later \citet{Henzinger1999} improved Aho et al.'s algorithm to run in \begin{math} min\{O(n + mn^{1/2}) , O(m + n^2 log n)\}\end{math} time. The time complexity of this algorithm further improved to \begin{math}min\{O(n + m log^2 n), O(m + n^2 log n)\}\end{math} by \citet{Jansson2005} using more recent data structures introduced by \citet{Holm2001}. MaxRTC is proved to be NP-hard \cite{Jansson2001, Wu2004, Bryant1997}. \citet{Byrka2008} reported that this proof is an L-reduction from an APX-hard problem meaning that the problem is APX-hard in general (non-dense case). Later, \citet{vanIersel2007} proved that MaxRTC is NP-hard even if the input triplet set is dense.

Several heuristics and approximation algorithms have been presented for the so called MaxRTC problem each of which performs better in practice on different input triplet sets. \citet{Gasieniec1999} proposed two algorithms by modifying Aho et al.'s algorithm. Their first algorithm which is referred as \texttt{One-Leaf-Split} \cite{Byrka2010} runs in \begin{math}O((m + n)log n)\end{math} time and the second one which is referred as \texttt{Min-Cut-Split} \cite{Byrka2010} runs in \begin{math}min\{O(mn^2 + n^3 log n), O(n^4)\}\end{math} time. The tree generated by the first algorithm is guaranteed to be consistent with at least one third of the input triplet set. This gives a lower bound for the problem. In another study, \citet{Wu2004} introduced a bottom up heuristic approach called \textbf{BPMF}\footnote{Best Pair Merge First} which runs in \begin{math}O(mn^3)\end{math} time. In the same study he proposed an exact exponential algorithm for the problem which runs in \begin{math}O((m + n^2) 3^n)\end{math} time and \begin{math}O(2^n)\end{math} space. According to the results of \citet{Wu2004} BPMF seems to perform well in average on randomly generated data. Later \citet{Maemura2007} presented a modified version of BPMF called \textbf{BPMR}\footnote{Best Pair Merge with Reconstruction} which employs the same approach but with a little different reconstruction routine. BPMR runs in \begin{math}O(mn^3)\end{math} time and according to Maemura et al.'s experiments outperforms BPMF. \citet{Byrka2008} designed a modified version of BPMF to achieve an approximation ratio of 3. They also investigated how \emph{MinRTI}\footnote{Minimum Rooted Triplet Inconsistency} can be used to approximate MaxRTC and proved that MaxRTC admits a polynomial-time \begin{math}(3 - \frac{2}{n-2})-\end{math}approximation.
\section{Algorithms and experimental results}
\label{sec:4}

In this section we present two new heuristic algorithms for the MaxRTC problem.
\subsection{FastTree}
\label{sec:4-1}

The first heuristic algorithm has a bottom up greedy approach which is faster than the other previously known algorithms employing a simple data structure.

Let \emph{R(T)} denote the set of all triplets consistent with a given tree \emph{T}. \emph{R(T)} is called the \emph{reflective triplet set} of \emph{T}. It forms a minimal dense triplet set and represents \emph{T} uniquely\cite{Jansson2006}. Now we define the \emph{closeness} of the pair \emph{\{i,j\}}. The closeness of the pair \emph{\{i,j\}}, \begin{math}C_{i,j}\end{math}, is defined as the number of triplets of the form \emph{ij\textbar k} in a triplet set. Clearly, for any arbitrary tree \emph{T}, closeness of cherry species equals \begin{math}n-2\end{math} which is maximum in \emph{R(T)}. The reason is that every cherry species has a triplet with every other specie. Now suppose we contract every cherry species of the form \{\emph{i,j}\} to their parents \begin{math}p_{ij}\end{math} and then update \emph{R(T)} as following. For each contracted cherry species \{\emph{i,j}\} we remove triplets of the form \emph{ij}\textbar\emph{k} from \emph{R(T)} and replace \emph{i} and \emph{j} with \begin{math}p_{ij}\end{math} within the remaining triplets. The updated set, \begin{math}R^\prime(T^\prime)\end{math}, would be the reflective triplet set for the new tree \begin{math}T^\prime\end{math}. Observe that for cherries of the form \begin{math}\{p_{ij}, k\}\end{math} in \begin{math}T^\prime\end{math}, \begin{math}C_{i,k}\end{math} and \begin{math}C_{j,k}\end{math} would equal n-3 in \emph{R(T)}. Similarly, for cherries of the form \begin{math}\{p_{ij}, p_{kl}\}\end{math} in \begin{math}T^\prime\end{math}, \begin{math}C_{i,k}\end{math}, \begin{math}C_{j,k}\end{math}, \begin{math}C_{i,l}\end{math} and \begin{math}C_{j,l}\end{math} would equal n-4 in \emph{R(T)}. This forms the main idea of the first heuristic algorithm. We first compute the closeness of pairs of species by visiting triplets. Furthermore, sorting the pairs according to their closeness gives us the reconstruction order of the tree. This routine outputs the unique tree \emph{T} for any given reflective triplet set \emph{R(T)}. Yet, we have to consider that the input triplet set is not always a reflective triplet set. Consequently, the reconstruction order produced by sorting may not be the right order. However, if the loss of triplets admits a uniform distribution it won't affect the reconstruction order. An approximate solution for this problem is refining the closeness. This can be done by reducing the closeness of the pairs \{\emph{i,k}\} and \{\emph{j,k}\} for any visited triplet of the form \emph{ij}\textbar\emph{k}. Thus, if the pair \{\emph{i,j}\} were actually cherries, then the probability of choosing the pairs \{\emph{i,k}\} or \{\emph{j,k}\} before choosing the pair \{\emph{i,j}\} due to triplet loss will be reduced. We call this algorithm FastTree. See Alg. ~\ref{alg:FastTree} for the whole algorithm.
\begin{algorithm}
\caption{FastTree}
\label{alg:FastTree}
\begin{algorithmic}[1]
\STATE{Initialize a forest \emph{F} consisting of \emph{n} one-node trees labeled by species.}
\FORALL{triplet of the form \emph{ij}\textbar\emph{k}}
\STATE {\begin{math}C_{i,j}\end{math}:=\begin{math}C_{i,j}\end{math}+1}
\STATE {\begin{math}C_{i,k}\end{math}:=\begin{math}C_{i,k}\end{math}-1}
\STATE {\begin{math}C_{j,k}\end{math}:=\begin{math}C_{j,k}\end{math}-1}
\ENDFOR
\STATE{Create a list \emph{L} of pairs of species.}
\STATE{Sort L according to the refined closeness of pairs with a linear time sorting algorithm.}
\WHILE{\textbar\emph{L}\textbar\begin{math}>\end{math}0}
\STATE{Remove the pair \{\emph{i,j}\} with maximum \begin{math}C_{i,j}\end{math}.}
\IF{\emph{i} and \emph{j} are not in the same tree}
\STATE{Add a new node and connect it to roots of trees containing \emph{i} and \emph{j}.}
\ENDIF
\ENDWHILE
\IF{\emph{F} has more than one tree}
\STATE{Merge trees in any order until there would be only one tree.}
\ENDIF
\RETURN{the tree in \emph{F}}
\end{algorithmic}
\end{algorithm}
\begin{theorem}
FastTree runs in \begin{math}O(m + \alpha(n)n^2)\end{math} time.
\end{theorem}
\begin{proof}
Initializing a forest in step 1 takes \begin{math}O(n)\end{math} time. Steps 2-6 take \begin{math}O(m)\end{math} time. We know that the closeness is an integer value between \begin{math}0\end{math} and \begin{math}n-2\end{math}. Thus, we can employ a linear time sorting algorithm \cite{Cormen1990}. There are \begin{math}O(n^2)\end{math} possible pairs, therefore, step 8 takes \begin{math}O(n^2)\end{math} time. Similarly, the while loop in step 9 takes \begin{math}O(n^2)\end{math} time. Each removal in step 10 can be done in \begin{math}O(1)\end{math} time. By employing optimal data structures which are used for disjoint-set unions\cite{Cormen1990}, the amortized time complexity of steps 11 and 12 will be \begin{math}O(\alpha(n))\end{math}, where \begin{math}\alpha(n)\end{math} is the inverse of the function \begin{math}f(x)=A(n,n)\end{math}, and \emph{A} is the well known fast-growing \emph{Ackermann} function. Furthermore, step 16 takes \begin{math}O(n\alpha(n))\end{math} time. Hence, the running time of FastTree would be \begin{math}O(m + \alpha(n)n^2)\end{math}.
\end{proof}
\begin{figure}
  \begin{center}
    \includegraphics[scale=0.5]{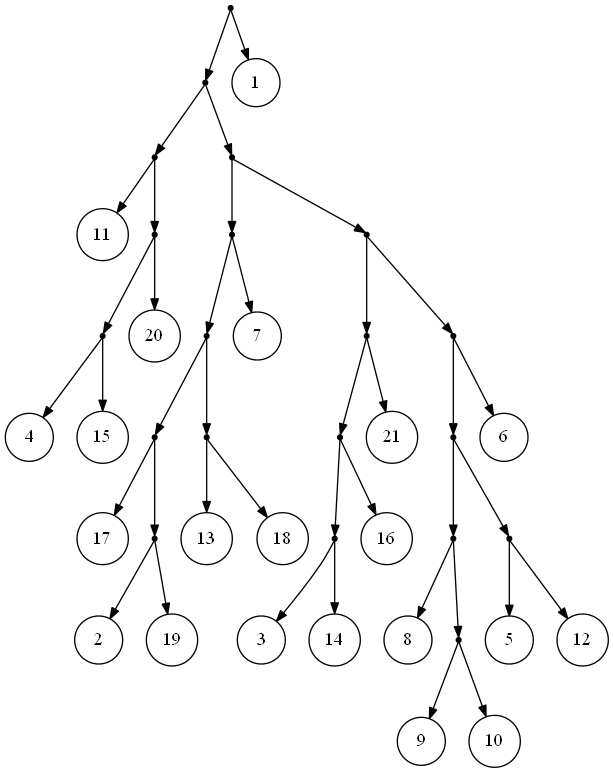}
  \end{center}
  \caption{Output of FastTree for a dense triplet set of yeast \emph{Cryptococcus gattii} data}
  \label{fig:2}
\end{figure}
Since \begin{math}A(4,4)=2^{2^{2^{65536}}}\end{math}, \begin{math}\alpha(n)\end{math} is less than 4 for any practical input size \emph{n}. In comparison to the fast version of Aho et al.'s algorithm FastTree employs a simpler data structure and in comparison to Aho et al.'s original algorithm it has smaller time complexity. Yet, the most important advantage of FastTree to Aho et al.'s algorithm is that it won't stuck if there is not a consistent tree with the input triplets, and it will output a proper tree in such a way that the clusters are very similar to that of the real network. The tree in Fig.~\ref{fig:2} is the output of FastTree on a dense set of triplets based on yeast \emph{Cryptococcus gattii} data. There is no consistent tree with the whole triplet set, however, \citet{vanIersel2008} presented a level-2 network consistent with the set(see Fig.~\ref{fig:3}). This set is available online \cite{triplets}. In comparison to BPMR and BPMF, FastTree runs much faster for large set of triplets and species. However, for highly sparse triplet sets, the output of FastTree may satisfy considerably less triplets than the tree constructed by BPMF or BPMR.
\begin{figure}
  \begin{center}
    \includegraphics[scale=0.75]{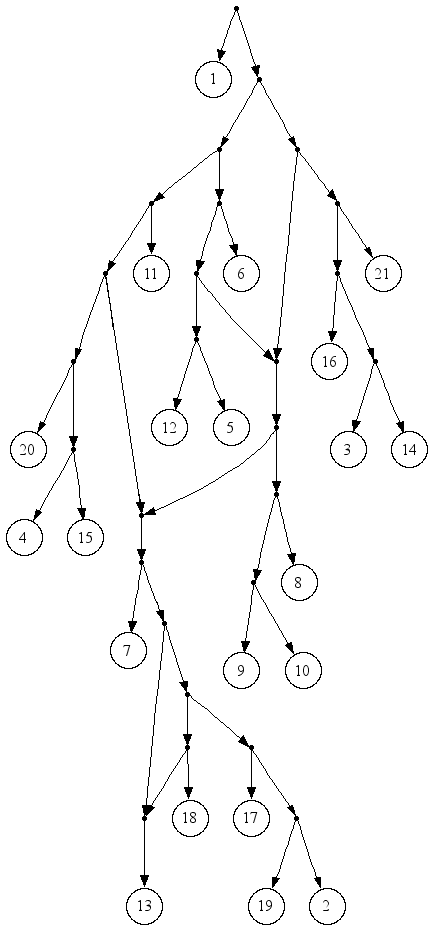}
  \end{center}
  \caption{A Level-2 network for dense triplet set of yeast \emph{Cryptococcus gattii} data}
  \label{fig:3}
\end{figure}

\subsection{BPMTR}
\label{sec:4-2}

Before explaining the second heuristic algorithm we need to survey BPMF \cite{Wu2004} and BPMR \cite{Maemura2007}.
BPMF utilizes a bottom up approach similar to hierarchal clustering. Initially, there are \emph{n} trees each of which contains a single node representing one of \emph{n} given species. In each iteration, the algorithm computes a function called \begin{math}e\_score \end{math} for each combination of two trees. Furthermore, two trees with the maximum \begin{math}e\_score\end{math} are merged into a single tree by adding a new node as the common parent of the selected trees. \citet{Wu2004} introduced six alternatives for computing the \begin{math}e\_score\end{math} using combinations of \emph{w}, \emph{p} and \emph{t}. (see Tab.~\ref{tab:1}). Though, in each run one of the six alternatives must be used. In the function \begin{math} e\_score(C_1,C_2)\end{math}, \emph{w} is the number of triplets satisfied by merging \begin{math}C_1\end{math} and \begin{math}C_2\end{math} which is the number of triplets of the form \emph{ij}\textbar\emph{k} in which \emph{i} is in \begin{math}C_1\end{math}, \emph{j} is in \begin{math}C_2 \end{math} and \emph{k} is neither in \begin{math}C_1\end{math} nor in \begin{math}C_2\end{math}. The value of \emph{p} is the number of triplets that is in conflict with merging \begin{math}C_1 \end{math} and \begin{math} C_2\end{math}. It is the number of triplets of the form \emph{ij}\textbar\emph{k} in which \emph{i} is in \begin{math} C_1\end{math}, \emph{k} is in \begin{math} C_2 \end{math} and \emph{j} is neither in \begin{math}C_1\end{math} nor in \begin{math}C_2\end{math}. The value of \emph{t} is the total number of triplets of the form \emph{ij}\textbar\emph{k} in which \emph{i} is in \begin{math}C_1 \end{math}and \emph{j} is \begin{math}C_2\end{math}. Wu compared the BPMF with \texttt{One-Leaf-Split} and \texttt{Min-Cut-Split} and showed that BPMF works better on randomly generated triplet sets. He also notifies that none of six alternatives of \begin{math}e\_score\end{math} is absolutely better than the other.\\
\begin{table}
\caption{The six alternatives of \emph{e\_score}}
\label{tab:1}       
\begin{tabular}{lllll}
If-Penalty
& & & Ratio Type\\
\hline\noalign{\smallskip}
False &  & w & w/(w+p) & w/t \\
True &  & w-p & (w-p)/(w+p) & (w-p)/t \\
\noalign{\smallskip}\hline\noalign{\smallskip}
\end{tabular}
\end{table}
\citet{Maemura2007} introduced a modified version of BPMF called BPMR outperforming the results of BPMF. BPMR works very similarly in comparison to BPMF except for a reconstruction step which is used in BPMR. Suppose \begin{math}T_x\end{math} and \begin{math}T_y\end{math} are two trees having the maximum \begin{math}e\_score \end{math} at some iteration and are selected to merge into a new tree. By merging \begin{math}T_x\end{math} and \begin{math}T_y\end{math} some triplets will be satisfied, but some other triplets will be in conflict. Without loss of generality, suppose \begin{math}T_x\end{math} has two subtrees namely left subtree and right subtree. Besides, suppose a triplet \emph{ij}\textbar\emph{k} in which \emph{i} is in the left subtree of \begin{math}T_x\end{math}, \emph{k} is in the right subtree of \begin{math}T_x\end{math} and \emph{j} is in \begin{math}T_y\end{math}. Observe that by merging \begin{math}T_x\end{math} and \begin{math}T_y\end{math} the mentioned triplet becomes inconsistent. However, swapping \begin{math}T_y\end{math} with the right subtree of the \begin{math}T_x\end{math} satisfies this triplet while some other triplets become inconsistent. It is possible that the resulting tree of this swap satisfy more triplets than the primary tree. This is the main idea behind the BPMR. In BPMR, in addition to the regular merging of \begin{math}T_x\end{math} and \begin{math}T_y\end{math}, \begin{math}T_y\end{math} is swapped with the left and the right subtree of \begin{math}T_x\end{math} and also \begin{math}T_x\end{math} is swapped with the left and the right tree of \begin{math}T_y\end{math}. Finally, among these five topologies we choose the one that satisfies more triplets.

Suppose the left subtree of the \begin{math}T_x\end{math} has also two subtrees. Swapping \begin{math}T_y\end{math} with one of these subtrees would probably satisfy new triplets while some old ones would become inconsistent. There are examples in which this swap results in a tree that satisfies more triplets. This forms our second heuristic idea that swapping of \begin{math}T_y\end{math} with every subtree of \begin{math}T_x\end{math} should be checked. \begin{math}T_x\end{math} should also be swapped with every subtree of \begin{math}T_y\end{math}. At every iteration of BPMF after choosing two trees maximizing the \begin{math}e\_score\end{math}, the algorithm tests every possible swapping of these two trees with subtrees of each other and then chooses the tree having the maximum consistency with triplets. We call this algorithm \emph{BPMTR}\footnote{Best Pair Merge with Total Reconstruction}. See Alg. ~\ref{alg:BPMTR} for details of the BPMTR.
\begin{algorithm}
\caption{BPMTR}
\label{alg:BPMTR}
\begin{algorithmic}[1]
\STATE{Initialize a set \emph{T} consisting of \emph{n} one-node trees labeled by species.}
\WHILE{\textbar\emph{T}\textbar\begin{math}>\end{math}1}
\STATE{Find and remove two trees $T_x$, $T_y$ with maximum $e\_score$.}
\STATE{Create a new tree $T_{merge}$ by adding a common parent to $T_x$ and $T_y$}
\STATE{$T_{best}$ := $T_{merge}$}
\FORALL{subtree $T_{sub}$ of $T_x$}
\STATE{Let $T_{swapped}$ be the tree constructed by swapping $T_{sub}$ with $T_y$}
\IF{the number of consistent triplets with $T_{swapped}$ was larger than the number of triplets consistent with $T_{best}$}
\STATE{$T_{best}$ := $T_{swapped}$}
\ENDIF
\ENDFOR
\FORALL{subtree $T_{sub}$ of $T_y$}
\STATE{Let $T_{swapped}$ be the tree constructed by swapping $T_{sub}$ with $T_x$}
\IF{the number of consistent triplets with $T_{swapped}$ was larger than the number of triplets consistent with $T_{best}$}
\STATE{$T_{best}$ := $T_{swapped}$}
\ENDIF
\ENDFOR
\STATE{Add $T_{best}$ to \emph{T}.}
\ENDWHILE
\RETURN{the tree in \emph{T}}
\end{algorithmic}
\end{algorithm}
\begin{theorem}
BPMTR runs in \begin{math} O(mn^3)\end{math} time.
\end{theorem}
\begin{proof}
Step 1 takes \begin{math}O(n)\end{math} time. In steps 2, initially \emph{T} contains \emph{n} clusters, but in each iteration two clusters merge into a cluster. Hence, the while loop in step 2 takes \begin{math}O(n)\end{math} time. In Step 3, \begin{math}e\_score\end{math} is computed for every subset of \emph{T} of size two. By applying Bender and Farach-Colton's preprocessing algorithm \cite{Bender2000} which runs in \begin{math}O(n)\end{math} time for a tree with \emph{n} nodes, every LCA query can be answered in \begin{math}O(1)\end{math} time. Therefore, the consistency of a triplet with a cluster can be checked in \begin{math}O(1)\end{math} time. Since there are \emph{m} triplets, step 3 takes \begin{math}{|T| \choose 2}O(m)\end{math} time. In steps 5, 9 and 15 \begin{math}T_{best}\end{math} is a pointer that stores the best topology found so far during each iteration of the while loop in \begin{math}O(1)\end{math} time. The complexity analysis of foreach loops in steps 6-11 and 12-17 are similar, and it is enough to consider one. Every rooted binary tree with \emph{n} leaves has \begin{math}O(n)\end{math} internal nodes so the total number of swaps in step 7 for any two clusters will be at most \begin{math}O(n-|T|)\end{math}. In step 8 computing the number of consistent triplets with \begin{math}T_{swapped}\end{math} takes no more than \begin{math}O(m)\end{math} time. Steps 4, 7 and 18 are implementable in \begin{math}O(1)\end{math} time. Accordingly, the running time of steps 2-19 would be:
\\
\begin{equation}\label{equ:1}
    \displaystyle\sum_{|T|=2}^n \Big[ m{|T| \choose 2} + O(n-|T|) + m) \Big]= O(mn^3)
\end{equation}
\\
Step 20 takes \begin{math}O(1)\end{math} time. Hence, the time complexity of BPMTR is \begin{math}O(mn^3)\end{math}.
\end{proof}
We tested BPMTR over randomly generated triplet sets with \emph{n} = 15, 20 species and \emph{m} = 500, 1000 triplets. We experimented hundred times for each combination of \emph{n} and \emph{m}. The results in Tab.~\ref{tab:2} indicate that BPMTR outperforms BPMR. However, in more than hundreds of tests there were few examples that BPMR performed better than BPMTR. For \emph{n}=30 and \emph{m}=1000, in sixty two triplet sets out of hundred randomly generated triplet sets, BPMTR satisfied more triplets. In thirty four triplet sets, BPMR and BPMTR had the same results and in only four triplet sets BPMR satisfied more triplets.

%
\begin{table}
\caption{Outperforming results of BPMTR in comparison to BPMR}
\label{tab:2}       
\begin{tabular}{lllll}
\\
\hline\noalign{\smallskip}
No. of species and triplets
& & \% better results & \% worse results\\
\hline\noalign{\smallskip}
n=20, m=500 & &  \%29 & \%0.0\\
n=20, m=1000 & & \%37 & \%1\\
n=30, m=500 & & \%61 & \%3\\
n=30, m=1000 & & \%62 & \%4\\
\noalign{\smallskip}\hline\noalign{\smallskip}
\end{tabular}
\end{table}

\section{Conclusion and Open Problems}
\label{sec:5}

In this paper we presented two new algorithms for the so called MaxRTC problem. For a given set of \emph{m} triplets on \emph{n} species, the FastTree algorithm runs in \begin{math}O(m + \alpha(n)n^2)\end{math} time which is faster than any other previously known algorithm, although, the outcome can be less satisfactory for highly spars triplet sets. The BPMTR algorithm runs in \begin{math}O(mn^3)\end{math} time and in average performs better than any other previously known approximation algorithm for this problem. There are still more ideas for improvement of the described algorithms.

1. In the FastTree algorithm to compute the closeness of pairs of species we check triplets, and for each triplet of the form \emph{ij}\textbar\emph{k} we add a weight \emph{w} to \begin{math}C_{i,j}\end{math} and subtract a penalty \emph{p} from \begin{math}C_{i,k}\end{math} and \begin{math}C_{j,k}\end{math}. In this paper, we set \begin{math}w=p=1\end{math}. If one assigns different values for \emph{w} and \emph{p} the closeness of pairs of species will be changed and the reconstruction order will be affected. It is interesting to check for which values of \emph{w} and \emph{p} FastTree performs better.

2. \citet{Wu2004} introduced six alternatives for \begin{math}e\_score\end{math} each of which performs better for different input triplet sets. It is interesting to find a new function outperforming all the alternatives for any input triplet set.

3. The best known approximation factor for the MaxRTC problem is 3 \cite{Byrka2008}. This is the approximation ratio of BPMF. Since MaxRTC is APX-hard a PTAS is unattainable, unless P=NP. However, \cite{Byrka2010} suggest that an approximation ratio in the region of 1.2 might be possible. Finding an \begin{math}\alpha-\end{math}approximation algorithm for MaxRTC with \begin{math}\alpha<3\end{math} is still open.

4. It is also interesting to find the approximation ratio of FastTree in general and for reflective triplet sets.
\section*{Acknowledgements}

The authors are grateful to thank Jesper Jansson and Fatemeh Zareh for reviewing this article, provding useful comments and answering our endless questions.



















\bibliographystyle{mdpi}
\makeatletter
\renewcommand\@biblabel[1]{#1. }
\makeatother
\bibliography{myrefs}





\end{document}